\documentclass{article}
\usepackage{import}
\usepackage{basic}

\title{Minimizing Inequity in Facility Location Games}
\date{}

\begin{document}

\author{Yuhang Guo\thanks{UNSW Sydney. Email: \texttt{yuhang.guo2@unsw.edu.au}.} \and Houyu Zhou\thanks{UNSW Sydney. Email: \texttt{houyu.zhou97@gmail.com}. Corresponding Author.}}

\maketitle

\begin{abstract}

This paper studies the problem of minimizing group-level inequity in facility location games on the real line, where agents belong to different groups and may act strategically. We explore a fairness-oriented objective that minimizes the maximum group effect introduced by Marsh and Schilling (1994). Each group's effect is defined as its total or maximum distance to the nearest facility, weighted by group-specific factors. We show that this formulation generalizes several prominent optimization objectives, including the classical utilitarian (social cost) and egalitarian (maximum cost) objectives, as well as two group-fair objectives, maximum total and average group cost. In order to minimize the maximum group effect, we first propose two novel mechanisms for the single-facility case, the \BALANCED mechanism and the \MajorPhantom mechanism. Both are strategyproof and achieve tight approximation guarantees under distinct formulations of the maximum group effect objective. Our mechanisms recover the known tight approximation guarantees for classical group-fair objectives as special cases, while extending them to a broader weighted framework and unifying several classical truthful mechanisms.

\end{abstract}

\thispagestyle{empty}

\vspace{15pt}

\hrule

\vspace{5pt}
{
	\setlength\columnsep{35pt}
	\setcounter{tocdepth}{2}
	\renewcommand\contentsname{\vspace{-20pt}}
		{\small\tableofcontents}
}

\vspace{11pt}
\hrule

\newpage
\pagenumbering{arabic}

\section{Introduction}\label{sec:intro}
Facility Location Games (FLGs), which study how to locate facilities based on agents' preferences, have been extensively explored over the past two decades. Most prior work in this area has prioritized efficiency, typically aiming to minimize the total cost incurred by agents in accessing services. Such efficiency-driven approaches achieve optimal social welfare, however, at the expense of fairness and equity. In particular, mechanisms designed purely for efficiency tend to favor majority groups, leaving disadvantaged or minority populations marginalized. 
Recognizing these limitations, recent research has increasingly focused on incorporating fairness into facility location games. These efforts span a spectrum from \emph{individual fairness}, which aims to equalize costs across agents~\citep{cai2016envy,Toby25}, to \emph{group fairness}, which ensures equitable treatment across predefined groups~\citep{MarshS94, ZLC22a, ALL+25a}. A seminal contribution by \citet{MarshS94} introduced various equity metrics, including the ``center'' objective, which seeks to \textbf{minimize the maximum group effect}. As they note,
\begin{quote}
    ``\textit{This is the earliest and most frequently used measure that has an equity component.}"
\end{quote}
This underscores the significance of the center objective in equity-aware location analysis. However, many subsequent studies have adopted a narrow interpretation of this objective, often modeling it as the maximum individual distance across all agents, thereby overlooking group structures and alternative definitions of group-level costs. In contrast, \citet{MarshS94} proposed a broader formulation, in which the effect of a group, $E_i$, could be defined in terms such as the \emph{total distance} incurred by all agents in group $i$. This generalization more accurately reflects the \emph{collective burden} borne by each group, offering a richer fairness perspective.
Motivated by this observation, we revisit the general objective of minimizing $\max_i E_i$ and study its implications within the framework of facility location games.

In this paper, we focus on the objective of minimizing the maximum group effect, where the effect $E_i$ of a group $i$ is defined as either the total or maximum distance from its agents to their nearest facility, multiplied by a weight $w_i$, capturing group-specific priorities, such as socioeconomic status or policy-driven importance (see \Cref{sec:model} for formal definitions). Our objective of minimizing $\max_i E_i$ emphasizes group-level fairness by bounding the worst-case burden among all groups. Unlike traditional formulations that protect only the most distant individual, our model accounts for the collective experience of each group. This perspective aligns with Rawlsian principles~\citep{Rawls1958}, which advocate prioritizing the welfare of the most disadvantaged. We aim to design new strategyproof mechanisms to this group-centric fairness objective.

\subsection{Related Work}
FLGs have received significant attention in the literature over the past decades, particularly following the influential work of~\citet{ProcacciaT09}. For an overview of the diverse models, we refer readers to the comprehensive survey by \citet{CFL+21a}. In the remainder of this section, we focus specifically on research that investigates fairness notions within the context of facility location games.

There is a rich body of work studying fairness considerations in facility location problems from the optimization perspective. Early studies in the \textit{operations research} community explored various equity-based fairness measures, including the standard deviation of distances~\citep{McAl76a} and the Gini coefficient~\citep{MarshS94}. In the context of algorithmic mechanism design, the pioneering work of~\citet{ProcacciaT09} introduced the notion of individual fairness through the maximum cost objective, i.e., minimizing the maximum individual distance from any agent to the facility, and proposed strategyproof mechanisms that approximately optimize this objective. Building on this foundation, later research proposed alternative formulations of individual fairness. \citet{cai2016envy} studied the minimax envy objective, which captures fairness via the maximum difference in distances between any pair of agents; this framework was subsequently extended to the two-facility setting by \citet{CFL+22a}. \citet{DLC+20a} introduced the envy ratio objective, adapted from the fair division problem, which measures the ratio of the best-off agent's utility to that of the worst-off agent. This concept was further extended to the multi-facility case by \citet{liu2020envyratio}. \citet{Toby25} studied the Gini index objective and proposed strategyproof mechanisms.

Fairness notions have also been extended from individuals to groups of agents. \citet{ZLC22a, LiLC24a} investigated two group-based fairness objectives: the \textit{maximum total cost} ($\mtgc$) and \textit{maximum average cost} ($\magc$), each capturing the worst-case burden across predefined groups of agents. \citet{ZhouCL24} studied these group-fair objectives in a setting with altruistic agents. \citet{AzizCSWX25} extended these objectives to a multi-period mobile facility location setting, where agents from different groups arrive over time and multiple facilities are located in each period; in the single-facility, single-period case, they obtained tight $2$-approximation bounds for both objectives, improving the previous approximation guarantee of $3$~\citep{ZLC22a}.
\citet{LiPWZ25} further incorporated group externalities into this framework, considering both competitive and collaborative interactions among group members. For the competitive setting, they obtained tight $2$-approximation bounds for both
$\mtgc$ and $\magc$.
\citet{ALL+25a} introduced a model of proportional fairness, in which fairness guarantees are provided to endogenously defined groups of agents, and the strength of the guarantee scales proportionally with group size. This concept was further extended by \citet{LAL+24a} to the setting of obnoxious facility location, where the facility imposes disutility rather than providing benefit.


\paragraph{Roadmap} Section~2 introduces the group-based facility location game model, formally defines our key objective, maximum group effect, and provides an overview of our main contributions. Section~3 focuses on the single-facility setting, where we propose two novel mechanisms tailored to distinct formulations of the maximum group effect objective. In Section~4, we extend our analysis to the multi-facility setting and revisit the classical \ENDPOINT mechanism within the framework. Due to space constraints, some proofs are omitted.

\section{Model and Contributions}\label{sec:model}
\subsection{Facility Location Games}
For any $t\in \mathbb{N}$, let $[t]:=\{1,2,\dots, t\}$. 
A facility location game consists of a set $N=[n]$ of $n$ agents, belonging to $m$ groups. For each agent $i \in N$, her type is denoted as $\theta_i=(x_i, g_i)$ where $x_i \in \mathbb{R}$ is the agent's \textit{private} location on a line, and $g_i \subseteq [m]$ denotes the set of their \textit{public} group memberships. The type profile of all agents is denoted by $\bftheta=(\theta_1,\theta_2,\dots, \theta_n)$.
Without loss of generality, we assume that agents are indexed by $x_1 \leq x_2 \leq \dots \leq x_n$. Let $G$ be the set of groups, $\mathcal{G}=\myset{G_1,G_2,\dots, G_m}$ where $G_j=\myset{i\in N : j \in g_i}$ represents the set of agents belonging to group $j$. Denote by $|G_j|$ the cardinality of $G_j$.  For each group $j \in [m]$, $j$ is assigned with a weight $w_j \geq 0$, reflecting their priority, such as socioeconomic factors or policy-driven importance. To simplify notation, let $w_{\max}=\max_{j \in [m]} w_j$ and $w_{\min}=\min_{j \in [m]} w_j$ denote the maximum and minimum group weights, respectively, and let $w_{g_i}=\max_{g\in g_i}\{w_{g}\}$ denote the maximum weight among all the groups to which agent $i$ belongs.
A deterministic mechanisms $f:\Theta^n \to \mathbb{R}^k$ maps the type profile $\bftheta$ to locations of $k$ facilities on a real line. Given any mechanism $f$, for each agent $i$, the cost incurred by $i$ is defined as $c(f(\bftheta), x_i) = \min_{y \in f(\bftheta)} |y - x_i|$, i.e., the distance from $x_i$ to the nearest facility. 

In this paper, we primarily focus on designing \textit{strategyproof} mechanisms. A mechanism $f$ is said to satisfy \textbf{strategyproofness} (SP) if it is in the best interests of every agent $i$ to report their truthful location $x_i$, irrespectively of the reports of the other agents. 
\begin{definition}[Strategyproofness (SP)]
A mechanism $f$ is strategyproof if, for any agent $i$ with true location $x_i$ and group $g_i$, any misreported location $x_i' \in \mathbb{R}$, and any profile $\bftheta'_{-i}$ of other agents' reports, we have:
\begin{align*}
c(f((x_i, g_i), \bftheta'_{-i}), x_i) \leq c(f((x_i', g_i), \bftheta'_{-i}), x_i).
\end{align*}
\end{definition}

\subsection{Maximum Group Effect}\label{subsec:mge}
Given the requirement of strategyproofness, our goal is to design mechanisms that minimize inequity, as captured by the maximum group effect objective, aligning with the central notion proposed by \citet{MarshS94}. 
Formally, for any profile $\bftheta$, the objective is to minimize the \textbf{maximum group effect} ($\mge$), defined as 
\begin{align*}
    \mge(\bftheta, f(\bftheta))=\max_{j\in [m]} E_j.
\end{align*}
We further interpret the \textit{maximum group effect} $E_j$ for each group $j$ in both \textit{utilitarian} and \textit{egalitarian} manners, incorporating the group-specific weight $w_j$ in both formulations. Specifically, one is termed  \textbf{Weighted Total Group Cost} ($\WTGC$), 
$E_j = w_j \cdot \sum_{i \in G_j} c(f(\bftheta), x_i)$, corresponding to the weighted sum of distances from all agents in group $G_j$ to their assigned facilities. The other is termed \textbf{Weighted Maximum Group Cost} ($\WMGC$), i.e., $ E_j = w_j \cdot \max_{i \in G_j} c(f(\bftheta), x_i)$, representing the weighted maximum distance among agents in group $G_j$. $\mge$ prioritizes fairness by limiting the worst-case weighted burden among groups, aligning with principles of equitable resource allocation.

For any strategyproof mechanism, we evaluate the performance by the \textbf{approximation ratio}, defined as the worst-case ratio (over all possible instances) between the maximum group effect produced by the mechanism and the optimal solution.
\begin{definition}[Approximation Ratio]
For any mechanism $f$, the approximation ratio is:
\begin{align*}
\rho = \sup_{\bftheta \in \Theta^n} \frac{\mge(\bftheta, f(\bftheta))}{\mge(\bftheta, \OPT(\bftheta))},
\end{align*}
where $\OPT(\bftheta)$ is the optimal facility placement that minimizes $\mge$ objective under the profile $\bftheta$.
\end{definition}

In the following sections, we focus on deterministic strategyproof mechanisms that optimize the $\mge$ objective. We begin with the single-facility setting ($k = 1$) and then extend our analysis to the multi-facility scenario.

\subsection{Our Contribution}\label{subsec:contribution}
We advance the field of fair mechanism design in FLGs by introducing a unified framework that seamlessly integrates efficiency, individual fairness, and group fairness. 

\paragraph{Capturing Fairness Through Generalized Metrics.}
We first introduce the general metric termed maximum group effect $(\mge)$, which is defined as $\mge = \max_{j \in [m]} E_j$. Here $E_j$ is interpreted either the \textit{weighted total group cost} ($\WTGC$), i.e., $E_j=w_j \cdot \sum_{i \in G_j} c(f(\bftheta), x_i)$, or the \textit{weighted maximum group cost} ($\WMGC$), i.e., $E_j=w_j \cdot \max_{i \in G_j} c(f(\bftheta), x_i)$. Our proposed $\mge$ objective offers a unified framework for analyzing efficiency and fairness in facility location problems, which captures a broad range of objectives by appropriately adjusting the group partitioning and weight assignments, as illustrated in \Cref{fig:mge_generalization_metric}.

\begin{figure}[!htbp]
    \centering
    \begin{tikzpicture}[
        block/.style={draw, thick, minimum width=2cm, minimum height=0.8cm, align=center},
        >=Stealth 
    ]
        \node[block] (O) {$\mge$};

        \node[block, below left=1cm and 1.5cm of O] (A) {\footnotesize \WTGC};
        \node[block, below right=1cm and 1.5cm of O] (F) {\footnotesize \WMGC};

        \node[block, below=2cm of A] (C) {$\magc$}; %
        \node[block, left=0.8cm of C] (B) {$\mtgc$};
        \node[block, right=0.8cm of C] (D) {Social Cost};
        \node[block, right=0.8cm of D] (E) {Max Cost};

        \draw[->] (O) -- (A);
        \draw[->] (O) -- (F);

        \draw[->] (A) -- (B) node[midway, left, xshift=-2pt] {\footnotesize $w_j=1$};
        \draw[->] (A) -- (C) node[near end, left] {\footnotesize $w_j=\frac{1}{|G_j|}$};
        \draw[->] (A) -- (D) node[near end, right] {\footnotesize $m=1$};
        \draw[->] (A) -- (E) node[midway, right, xshift=2pt] {\footnotesize $m=n$};  

        \draw[->] (F) -- (E) node[midway, right] {\footnotesize $m=1$}; 
    \end{tikzpicture}
    \caption{$\mge$: A Systematic Generalization of Objective Metrics.}
    \label{fig:mge_generalization_metric}
\end{figure}
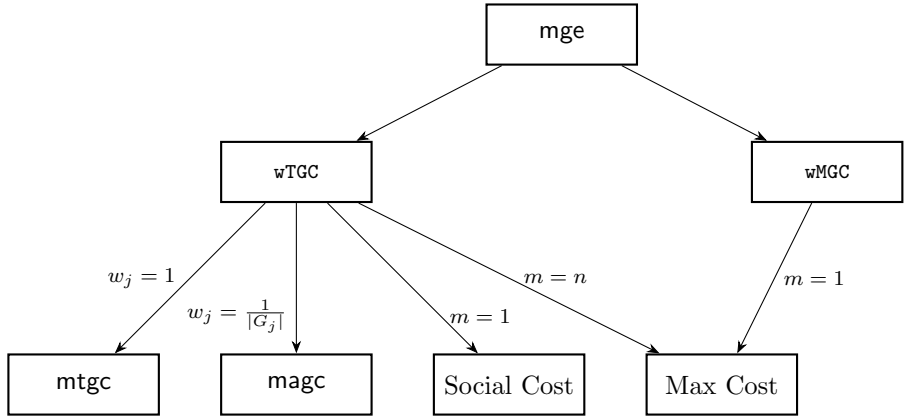

\paragraph{Unified Mechanisms with Tight Approximation Guarantee.}
We propose two novel strategyproof mechanisms: the \BALANCED mechanism and the \MajorPhantom mechanism. The \BALANCED mechanism, a flagship contribution for single-facility location games, not only unifies classic facility location mechanisms, but also provides tight results for group-fairness objectives. Specifically, regarding the social cost objective, the \BALANCED mechanism aligns with the median-point mechanism, while for  the maximum cost objective, it degenerates to the \LEFTMOST mechanism. 
In the context of group fairness, the \BALANCED mechanism specializes to a $2$-approximation for both the maximum total group cost ($\mtgc$) and maximum
average group cost ($\magc$), matching the known tight bounds. More generally, we prove that \BALANCED achieves a tight $2$-approximation for arbitrary
weighted total group costs, allowing overlapping group memberships.


Consequently, this unification establishes the \BALANCED mechanism as a versatile instrument capable of adapting to diverse fairness and efficiency goals without bespoke designs for objectives. To optimize the weighted maximum group cost objective, we introduce the \MajorPhantom mechanism and show it provides tight results for this objective. 

In the two-facility setting (see Section~\ref{sec:multi}), we revisit the \ENDPOINT mechanism~\citep{ProcacciaT09}, which places facilities at the leftmost and rightmost agent locations. For both $\WTGC$ and $\WMGC$, we show it achieves tight bounds. For settings beyond two facilities ($k > 2$), we leverage results from \citet{FotakisT14} to show that all strategyproof, anonymous, and deterministic mechanisms yield unbounded approximation ratios. Our established tight approximation ratios and matching lower bounds are comprehensively summarized in~\Cref{tab:summary_of_results}. 
\begin{table}[htb]
\centering
\resizebox{.6\columnwidth}{!}{
\begin{NiceTabular}{@{}cccc@{}}
\toprule
\textbf{Setting} & \textbf{Objectives} & \textbf{Mechanisms} & \textbf{Bounds} \\
\midrule
$k=1$ & $\WTGC$ & \cellcolor{gray!30}\BALANCED & \cellcolor{gray!30}$2$ \\
 & $\WMGC$ & \cellcolor{gray!30}\MajorPhantom & \cellcolor{gray!30}$2$ \\
$k=2$ & $\WTGC$ & \ENDPOINT & \cellcolor{gray!30}$1+(n-2)\frac{w_{\max}}{w_{\min}}$\\
 & $\WMGC$ & \ENDPOINT & \cellcolor{gray!30}$1+\frac{w_{\max}}{w_{\min}}$ \\
$k \geq 3$ & General & / & $\infty$ \\
\bottomrule
\end{NiceTabular}
}
\caption{Summary of results. All listed bounds are tight. Gray shading denotes the contributions of this work.}
\label{tab:summary_of_results}
\end{table}

\section{Single-Facility Mechanism Design}\label{sec:single}
We begin by considering single-facility setting ($k=1$). For the $\WTGC$ group effect metric, i.e.,  $E_j=w_j \cdot \sum_{i \in G_j} c(f(\bftheta), x_i)$, we propose the \BALANCED mechanism which places the facility at the location under which the maximum weighted values are balanced. Regarding the $\WMGC$ objective where $E_j=w_j \cdot \max_{i \in G_j} c(f(\bftheta), x_i)$, we introduce the \MajorPhantom mechanism, which places the facility at the median-point of the locations of agents in the group $G_{\max}$ with the largest weight, together with $(|G_{\max}|-1)$ constant phantom points. For both cases, we provide tightness results, showing the optimality of the proposed mechanisms.

\subsection{Weighted Total Group Cost}
We first look at the weighted total group cost ($\WTGC$) objective, where $E_j = w_j \cdot \sum_{i \in G_j} c(f(\bftheta), x_i)$. As illustrated in \Cref{fig:mge_generalization_metric}, the $\WTGC$ metric subsumes several well-studied objectives, including the social cost, maximum cost, and group-fairness objectives $\mtgc$ and $\magc$. To address this generalized setting, we introduce our first \BALANCED mechanism (Mechanism~\ref{algo:balanced}), which places the facility at a location that equilibrates the weighted number of agents on either side of it.

\begin{algorithm}[!htbp]
\caption{\BALANCED \textbf{Mechanism}}
\label{algo:balanced}
\begin{algorithmic}[1] 
\REQUIRE Agent profile $\bftheta$, group weights $\{w_j\}_{j\in [m]}$.
\STATE Define $L_j(y) \leftarrow \lvert \{ i\in N:x_i \leq y \text{ and } j \in g_i \} \rvert$ and $R_j(y) \leftarrow \lvert \{ i \in N: x_i > y \text{ and } j \in g_i \} \rvert$. 
\STATE Compute\begin{align*}
 f(\bftheta)\hspace{-.2em}\leftarrow \hspace{-.2em}\min\hspace{-.2em}
       \left\{\hspace{-.2em}y \hspace{-.2em}\in \mathbb{R} \hspace{-.2em}:\hspace{-.2em} \max_{j\in [m]}\! w_jL_j(y)\hspace{-.2em}\geq\hspace{-.2em} \max_{j'\in [m]}\!w_{j'}R_{j'}(y) \hspace{-.2em}\right\}\hspace{-.2em}.
    \end{align*}
\ENSURE Facility location $f(\bftheta)$.
\end{algorithmic}
\end{algorithm}
Intuitively, the \BALANCED mechanism defines, for each group $j \in [m]$ and location $y \in \mathbb{R}$, two functions $L_j(y)$ and $R_j(y)$, representing respectively the number of agents in group $j$ located at or to the left of $y$, and those located to its right. Since $w_jL_j(y)$ is non-decreasing and $w_jR_j(y)$ is non-increasing over $y \in [x_1, x_n]$, the mechanism places the facility at a location that most closely balances the two quantities $\max_{j \in [m]} w_j L_j(y)$ and $\max_{j \in [m]} w_j R_j(y)$. The \BALANCED mechanism can be implemented in $O((n+m) \log n)$ time by performing a binary search over the sorted agent locations to identify the smallest $x_i$ satisfying $\max_{j \in [m]} w_j L_j(x_i) \ge \max_{j' \in [m]} w_{j'} R_{j'}(x_i)$. Each iteration of the binary search requires evaluating these maximum functions, which takes $O(n+m)$ time.

\begin{proposition}
\label{prop:balance_mechanism_mid_point_and_left_most}
The \BALANCED mechanism coincides with the median-point mechanism (resp. leftmost mechanism) when the $\mge$ objective degenerates to the social cost (resp. maximum cost) objective.
\end{proposition}

We next present the main theorem for the \BALANCED mechanism, which satisfies strategyproofness and achieves a $2$-approximation for minimizing the $\mge$ objective when $E_j=w_j\cdot \sum_{i\in G_j}c(f(\bftheta),x_i)$ for all $j\in [m]$.
\begin{theorem}
    \label{thm:balanced_main_theorem}
    The \BALANCED mechanism is strategyproof and has an approximation ratio of $2$ for minimizing $\mge$ when $E_j = w_j \cdot \sum_{i \in G_j} c(f(\bftheta), x_i)$.
\end{theorem}
\begin{proof}
    \textbf{Strategyproofness}. Let $f(\bftheta)$ denote the outcome of the \BALANCED mechanism, and consider any agent $i$ with truthful location $x_i$. We prove by discussing the relative positions of $f(\bftheta)$ and $x_i$. Clearly, if $x_i$ coincides with $f(\bftheta)$, agent $i$ has no incentive to misreport her location.
    Case (1). If $x_i < f(\bftheta)$, misreporting $x_i' < f(\bftheta)$ will not change the facility location as $L_j(f(\bftheta))$ and $R_j(f(\bftheta))$ don't change for all $j\in [m]$. If agent $i$ misreports to $x_i' \geq f(\bftheta)$, we have that $L_j(f(\bftheta))$ decreases and $R_{j}(f(\bftheta))$ increases for each $j\in g_i$, potentially shifting the facility location rightward as $\max_{j' \in [m]} \{ w_{j'}R_{j'}(f(\bftheta)) \}$ increases while $\max_{j \in [m]} w_j\{ L_{j}(f(\bftheta)) \}$ decreases. Consequently, agent $i$’s cost increases as the facility moves farther from $x_i$, implying that misreporting cannot be beneficial. We next consider Case (2). If $x_i > f(\bftheta)$, similarly, when misreporting $x_i' > f(\bftheta)$, the facility location remains at $f(\bftheta)$. When $x_i' \leq f(\bftheta)$, by an analogical induction, it could only potentially pushing the facility location farther away from agent $i$'s location $x_i$. Hence, we conclude that for any agent $i \in N$, $i$ has no incentive to misreport her location, which implies the \BALANCED mechanism is strategyproof.

    \textbf{Approximation Ratio.} Denote by $f(\bftheta)$ the \BALANCED mechanism outcome and $y^*=\argmin_{y\in \mathbb{R}}\mge(\bftheta,y)$ the optimal location for profile $\bftheta$. 
    We begin with a key observation that underpins the proof of the approximation ratio. Given any profile $\bftheta$, we construct a modified profile $\bftheta'$ by relocating all agents whose positions lie between $f(\bftheta)$ and $y^*$ to the point $y^*$. Under the construction, we first observe that for each agent $i$ who lies between $f(\bftheta)$ and $y^*$, $c(f(\bftheta), x_i)$ increases while $c(y^*,x_i)$ decreases, which follows that $\mge(\bftheta',f(\bftheta)) \geq \mge(\bftheta, f(\bftheta))$ and $\mge(\bftheta',y^*) \leq \mge(\bftheta, y^*)$. Let $\rho(\bftheta)$ (resp. $\rho(\bftheta')$) denote the approximation ratio under $\bftheta$ (resp. $\bftheta'$). By construction, we have $\rho(\bftheta) \leq \rho(\bftheta')$. Henceforth, we focus exclusively on profiles involving such movements. For the sake of clarity, we will abuse the notation $\bftheta$ to refer to the modified profile. 

    \textbf{Case 1: $f(\bftheta)\leq y^*$.} For each group $G_j$, the group effects under $f(\bftheta)$ and $y^*$ are expressed as $E_j(f(\bftheta))=w_j\sum_{i\in G_j}|f(\bftheta)-x_i|$ and $E_j(y^*)=w_j\sum_{i\in G_j}|y^*-x_i|$. Viewing $E_j(y)$ as a function of location $y$. its derivative can be expressed as $\frac{\mathrm{d} E_j(y)}{\mathrm{d}y}=w_j\cdot (-L_j(y)+R_j(y))$.
    Consequently, we derive that 
    \begin{align*}
        E_j(f(\bftheta))-E_j(y^*)=\int_{f(\bftheta)}^{y^*}w_j\cdot (R_j(y)-L_j(y))\mathrm{d}y.
    \end{align*}
    Since there is no agent located in the interval $[f(\bftheta),y^*)$, the derivative is a constant value. We further have
    \begin{align*}
        E_j(f(\bftheta)) \!\! - \!\! E_j(y^*)\!=\!w_j(R_j(f(\bftheta)) \!\hspace{-.1em} - \hspace{-.1em}\! L_j(f(\bftheta))) \! \cdot \! (y^* \!\! - \!\! f(\bftheta)).
    \end{align*}
    Recall that $\mge(\bftheta,f(\bftheta))=\max_{j\in [m]}E_j(f(\bftheta))$ and $\mge(\bftheta,y^*)=\max_{j\in [m]}E_j(y^*)$. We then have 
    \begin{align*}
        \mge(\bftheta,f(\bftheta)) - \mge(\bftheta,y^*) 
        & \leq \max_{j\in [m]}\{E_j(f(\bftheta))-E_j(y^*)\} \\
        & \leq \max_{j\in [m]}\{w_j(R_j(f(\bftheta))-L_j(f(\bftheta)))(y^*-f(\bftheta))\}.
    \end{align*}
    On the other hand, there are $L_j(f(\bftheta))$ agents in each group $G_j$ who are at or on the left of $f(\bftheta)$. It follows that 
    \begin{align*}
        \mge(\bftheta,y^*)\geq \max_{j\in [m]}\{w_j\cdot L_j(f(\bftheta))\}\cdot (y^*-f(\bftheta)).
    \end{align*}
    Moreover, we have 
    $\max_{j\in [m]} \{w_j R_j(f(\bftheta))\} \leq \max_{j\in [m]} \{w_jL_j(f(\bftheta))\}$ as $f(\bftheta)$ is the outcome by \BALANCED mechanism.  
    With these inequalities in hand, we derive the approximation ratio $\rho$ of the \BALANCED mechanism.
    \begin{align*}
    \rho & =\hspace{-.2em}  \frac{\mge(\bftheta, f(\bftheta))}{\mge(\bftheta, y^*)} = 1 +  \frac{\mge(\bftheta, f(\bftheta))-\mge(\bftheta, y^*)}{\mge(\bftheta, y^*)}\\
    & \leq\hspace{-.2em} 1 \hspace{-.1em}+\hspace{-.1em} \frac{\max_{j \in [m]} \myset{ w_j (R_j(f(\bftheta)) \hspace{-.1em}- \hspace{-.1em}L_j(f(\bftheta)))} (y^* \hspace{-.1em}- \hspace{-.1em}f(\bftheta)) }{\max_{j \in [m]} \myset{ w_j L_j(f(\bftheta)) } (y^* \!-\! f(\bftheta))}\\
    &\leq \hspace{-.2em} 1 \hspace{-.1em}+\hspace{-.1em} \frac{\max_{j \in [m]} \{ w_j R_j(f(\bftheta)) \}}{\max_{j \in [m]} \{ w_j L_j(f(\bftheta)) \}} \leq 2.
    \end{align*}
    \textbf{Case 2: When $f(\bftheta) > y^*$.} Keep in mind that we still have the derivative expression for function $E_j(y)$. Observe that in the interval $[y^\ast,f(\bftheta))$, the derivative of function $E_j(y)$ is a constant value as there is no agent in the interval. By an analogous approach, we have 
    \begin{align*}
        E_j(f(\bftheta)) \hspace{-.1em} - \hspace{-.1em} E_j(y^*) \hspace{-.1em}=\hspace{-.1em}w_j(R_j(y^*)\hspace{-.1em}-\hspace{-.1em} L_j(y^*))\cdot (f(\bftheta)-y^*).
    \end{align*}
    We next establish the group effect difference between solution $f(\bftheta)$ and $y^*$.
    \begin{align*}
        \mge(\bftheta,f(\bftheta)) - \mge(\bftheta,y^*) 
        & \leq \max_{j\in [m]}\{E_j(f(\bftheta))-E_j(y^*)\} \\
        & \leq \max_{j\in [m]}\{w_j(R_j(y^*)-L_j(y^*))(y^*-f(\bftheta))\}.
    \end{align*}
    Since we know that there is no agent in the interval $(y^*, f(\bftheta))$, then for all the agents on the right of $y^*$, they must satisfy $x_i \geq f(\bftheta)$. Therefore, we can bound $\mge(\bftheta,y^*)$ by
    \begin{align*}
        \mge(\bftheta, y^*) \geq \max_{j\in [m]}\{w_j\cdot R_j(y^*)\}\cdot (f(\bftheta)-y^*).
    \end{align*}
    Based on the aforementioned analysis, we derive the approximation ratio 
    \begin{align*}
    \rho & = \frac{\mge(\bftheta, f(\bftheta))}{\mge(\bftheta, y^*)}\\
    & \leq 1 + \frac{\max_{j \in [m]} \myset{ w_j (R_j(y^*) - L_j(y^*))} (f(\bftheta)-y^*) }{\max_{j \in [m]} \myset{ w_j \cdot R_j(y^*) } (f(\bftheta)-y^*)}\\
    & \leq 1 + \frac{\max_{j \in [m]} \{ w_j R_j(y^*) \}}{\max_{j \in [m]} \{ w_j R_j(y^*) \}} = 2.
    \end{align*}
    Combining the analyses of both cases, we conclude that the \BALANCED mechanism achieves an approximation ratio of $2$ for minimizing the maximum group effect when $E_j = w_j \cdot \sum_{i \in G_j} c(f(\bftheta), x_i)$.
\end{proof}

Since $\mge$ generalizes the maximum cost objective when $m=n$ and $G_j = \{j\}$ with equal weights, the lower bound from \citet{ProcacciaT09} applies, confirming the tightness of \BALANCED’s approximation ratio.

\begin{corollary}[\citealp{ProcacciaT09}]
    Any deterministic mechanism has an approximation ratio of at least $2$ for $\mge$ when $E_j = w_j \cdot \sum_{i \in G_j} c(f(\bftheta), x_i)$.
\end{corollary}

Note that the $\mge$ objective also generalizes the group-fair objectives $\mtgc$ and $\magc$ proposed by \citet{ZLC22a}. While \citet{ZLC22a} established
a $3$-approximation upper bound and a lower bound of $2$. Subsequent work obtained tight 2-approximation guarantees for these objectives
in related extensions \citep{AzizCSWX25,LiPWZ25}. \Cref{thm:balanced_main_theorem} recovers this tight factor as a special case of our more general weighted formulation.

\begin{corollary}
    The \BALANCED mechanism achieves a $2$-approximation ratio w.r.t. the $\mtgc$ and $\magc$ objectives.
\end{corollary}

\subsection{Weighted Maximum Group Cost}
We next turn to the weighted maximum group cost objective, wherein $E_j = w_j \cdot \max_{i \in G_j} c(f(\bftheta), x_i)$. 
Intuitively, it is a weighted maximum cost problem where each agent $i$ is assigned with a maximum weight $w_{g_i}=\max_{j\in g_i}w_j$, and the objective is to minimize the maximum value of $w_{g_i} \cdot c(f(\bftheta), x_i)$ over all agents $i\in N$. In view of this, we propose the \MajorPhantom mechanism (Mechanism~\ref{algo:major_phantom}) which selects the facility location by prioritizing the group with the largest weight $w_j$. 

\begin{algorithm}[!htbp]
\caption{\MajorPhantom \textbf{Mechanism}}
\label{algo:major_phantom}
\begin{algorithmic}[1] 
\REQUIRE Agent profile $\bftheta$, group weights $\{w_j\}_{j\in [m]}$.
\STATE Let $G_{\max}$ denote the largest weight group and $\bfx^{G_{\max}}=\{x_1^{G_{\max}}, \dots, x_{|G_{\max
}|}^{G_{\max}}\}$ denote the location profile of agents in $G^*$, tie-breaking in favor of the smallest index.
\STATE Let $v_1 \leq \dots \leq v_{|G_{\max}|-1}$ denote $|G_{\max}|-1$ values $v_1 \leq \dots \leq v_{|G_{\max}|-1}$. 
\STATE $f(\bftheta)\leftarrow \medp(\bfx^{G_{\max}}, v_1,\dots,v_{|G_{\max}|-1})$, tie-breaking by selecting the leftmost.
\ENSURE Facility location $f(\bftheta)$.
\end{algorithmic}
\end{algorithm}
\MajorPhantom mechanism extends the \PHANTOM mechanisms \citep{Moulin80a} by prioritizing agents in the largest-weighted group, thereby ensuring fairness for groups with greater importance. Before analyzing the approximation ratio of the \MajorPhantom mechanism with respect to the $\WMGC$ objective, we first provide a characterization of the optimal solution in two-agent instances, which will facilitate the subsequent analysis.

\begin{lemma}\label{lem:two_agents}
Given two-agent profile $\bftheta$ and optimal solution $y^*$, for any two agents with locations $x_1 \leq x_2$ and maximum weights $w_{g_1}, w_{g_2} \geq 0$, $\mge(\bftheta, y^*)=\max_{j \in [m]} w_j \cdot \max_{i \in G_j} |y^* - x_i|$ is either
\begin{itemize}
    \item $w_{g_1} \cdot \frac{(x_2 - x_1)}{2}$ when $g_1=g_2$ with the maximum weight $w_{g_1}$, achieved at $y^*=\frac{x_1+x_2}{2}$; or 
    \item $\frac{w_{g_1}\cdot w_{g_2} \cdot (x_2 - x_1)}{w_{g_1} + w_{g_2}}$ when $w_{g_1}\neq w_{g_2}$, achieved at $y^* = \frac{w_{g_2} x_2 + w_{g_1} x_1}{w_{g_1} + w_{g_2}}$.
\end{itemize}
\end{lemma}

We next prove that for any \MajorPhantom mechanism, it is strategyproof and achieves an approximation ratio of $2$ for minimizing $\mge$ under the $\WMGC$ objective.
\begin{theorem}
    Any \MajorPhantom mechanism is strategyproof and has an approximation ratio of $2$ for minimizing the $\mge$ objective when $E_j=w_j\cdot \max_{i\in G_j} c(f(\bftheta),x_i)$.
\end{theorem}
\begin{proof}
    \textbf{Strategyproofness}. Given any agent profile $\bftheta$, consider an agent $i \in N$ with true location $x_i$. If $i \notin G_{\max}$, it is clear that misreporting cannot influence the facility location under the mechanism. If $i \in G_{\max}$, then since group membership cannot be misreported, we can apply a similar analytical approach to that used in the proof of strategyproofness for \PHANTOM mechanisms by \citet{Moulin80a}.

    \textbf{Approximation Ratio.} Given any profile $\bftheta$, let $f(\bftheta)$ denote the location outputted by the \MajorPhantom mechanism and $y^*$ denote the optimal location under $\bftheta$. 

   We first consider the case that $y^\ast \geq f(\bftheta)$. Suppose that $\mge(\bftheta,f(\bftheta))$ is achieved by agent $\ell$. We first observe that if $x_\ell \leq f(\bftheta)$, we have $\mge(\bftheta, f(\bftheta))=w_{g_\ell}\cdot (f(\bftheta)-x_\ell)$ and $\mge(\bftheta,y^*)\geq w_{g_\ell}\cdot (y^*-x_\ell)$, which gives us
    \begin{align*}
        \rho = \frac{\mge(\bftheta, f(\bftheta))}{\mge(\bftheta, y^*)} \le \frac{w_{g_\ell}\cdot (f(\bftheta)-x_\ell)}{w_{g_\ell}\cdot (y^*-x_\ell)} \le 1.
    \end{align*}
If $x_\ell > f(\bftheta)$, let $k \in G_{\max}$ be the agent in group $G_{\max}$ whose location $x_k$ satisfies $x_k = \min_{j \in G_{\max}} {|x_j - f(\bftheta)|}$ with the additional condition that $x_k \leq f(\bftheta)$. That is, $x_k$ is the closest agent to the left of $f(\bftheta)$ within $G_{\max}$\footnote{If multiple agents satisfy this condition, we break ties by selecting the agent with the largest index $k$.}. We claim that such an agent $k$ always exists. Toward this end, suppose, for the sake of contradiction, that no such $x_k$ exists. It implies that all the agents' locations $\{x_1^{G_{\max}}, \dots, x_{|G_{\max}|}^{G_{\max}}\}$ lie strictly to the right of $f(\bftheta)$, i.e., $x_i^{G_{\max}} > f(\bftheta)$ for all $i \in G_{\max}$. However, under the \MajorPhantom mechanism, the facility is placed at the median of the multiset $\{\bfx^{G_{\max}}, v_1, \dots, v_{|G_{\max}|-1} \}$, which has a size of $2\cdot |G_{\max}| - 1$. In this case, there can be at most $|G_{\max}| - 1$ points strictly to the right of $f(\bftheta)$, contradicting the assumption. Therefore, such an agent $k$ always exists. 

If $x_k\leq f(\bftheta) < x_\ell < y^*$, we have $\mge(\bftheta,f(\bftheta))=w_{g_\ell}\cdot (x_\ell - f(\bftheta))\leq w_{g_\ell}\cdot (x_\ell-x_k)$ and $\mge(\bftheta,y^*)\geq w_{g_k}\cdot (y^*-x_k)\geq w_{g_k}\cdot (x_\ell-x_k)$. Hence, the approximation ratio is expressed as 
\begin{align*}
    \rho = \frac{\mge(\bftheta, f(\bftheta))}{\mge(\bftheta, y^*)} \le \frac{w_{g_\ell}\cdot (x_\ell-x_k)}{w_{g_k}\cdot (x_\ell-x_k)} = \frac{w_{g_\ell}}{w_{g_k}}.    
\end{align*}
Recall the definition of \MajorPhantom mechanism. We know that $\rho \leq \frac{w_{g_\ell}}{w_{g_k}} =\frac{w_{g_\ell}}{w_{\max}}\leq 1$ as $w_{g_k}=w_{\max}\geq w_{g_\ell}$.

If $x_k\leq f(\bftheta) < y^* < x_\ell$, we have $\mge(\bftheta, f(\bftheta))=w_{g_\ell}\cdot (x_{\ell}-f(\bftheta)) \leq w_{g_\ell}\cdot (x_\ell - x_k)$. By \Cref{lem:two_agents}, when only considering agent $k$ and $\ell$, we have the maximum cost achieved by these two agents is at least $\frac{w_{g_\ell} w_{g_k} (x_\ell-x_k)}{w_{g_\ell} + w_{g_k}}$. Hence, we have
$\mge(\bftheta, y^*) \ge \frac{w_{g_\ell} w_{g_k} (x_\ell-x_k)}{w_{g_\ell} + w_{g_k}}$.
Consequently, the approximation ratio is bounded by 
\begin{align*}
    \rho = \frac{\mge(\bftheta, f(\bftheta))}{\mge(\bftheta, y^*)} \le \frac{w_{g_\ell}\cdot (x_\ell-x_k)}{\frac{w_{g_\ell} w_{g_k} (x_\ell-x_k)}{w_{g_\ell} + w_{g_k}}} = \frac{w_{g_\ell} + w_{g_k}}{w_{g_k}}.
\end{align*}
Since $w_{g_k} = w_{\max} \geq w_{g_\ell}$, it follows that $ \rho = \frac{w_{g_\ell} + w_{g_k}}{w_{g_k}} = \frac{w_{g_\ell} + w_{\max}}{w_{\max}} \leq 2$.
For the case where $y^* < f(\bftheta)$, the same approximation ratio of $2$ can be established by applying an analogous analysis to that used for the case $f(\bftheta) \geq y^*$.
\end{proof}

Notice that the $\mge$ objective coincides with the maximum cost objective when $m=n$ and $G_j = \{j\}$ with equal weights. In this case, the lower bound of $2$ for the maximum cost objective established by \citet{ProcacciaT09} applies, thereby confirming the tightness of the bounds achieved by the \MajorPhantom mechanism. 
\begin{corollary}[\citealp{ProcacciaT09}]
    Any deterministic, strategyproof mechanism has an approximation ratio of at least $2$ for $\mge$ when $E_j = w_j \cdot \max_{i \in G_j} c(f(\bftheta), x_i)$.
\end{corollary}

\section{Multi-Facility Mechanism Analysis}\label{sec:multi}
In this section, we extend our analysis from single-facility to multi-facility settings. In view of the impossibility result of \citet{FotakisT14}, which shows that for $k \geq 3$, no deterministic, anonymous, and strategyproof mechanism can achieve a bounded approximation ratio for either the social cost or maximum cost objectives, our primary focus is on the two-facility case ($k=2$).

\begin{corollary}
    When $k\geq 3$, there is no deterministic, anonymous, strategyproof mechanisms with a bounded approximation ratio for $\mge$, for either $E_j = w_j \cdot \sum_{i \in G_j} c(f(\bftheta), x_i)$ or $E_j = w_j \cdot \max_{i \in G_j} c(f(\bftheta), x_i)$.
\end{corollary}

We next restrict our attention to the case of $k=2$ and revisit the \ENDPOINT mechanism (placing facilities at the leftmost and rightmost agent locations), which remains the only known deterministic, anonymous, and strategyproof mechanism with bounded approximation guarantees for these objectives \citep{FotakisT14}.

While \citet{FotakisT14} established that the \ENDPOINT mechanism is the only deterministic, anonymous, and strategyproof mechanism with bounded approximation guarantees for social cost in the two-facility setting ($k=2$), evaluating its performance under our group-centric $\mge$ objective presents a novel and nontrivial challenge. Unlike classical objectives, the $\mge$ objective requires accounting for weighted group effects, where both the group structures and the distribution of group weights play a critical role, which demand a fundamentally different analytical approach. Our contribution lies in establishing tight approximation bounds for the \ENDPOINT mechanism under $\mge$, thereby extending its applicability to equitable facility placement and offering theoretical insights into group fairness.

\subsection{Weighted Total Group Cost}
We first explore the maximum group effect objective by considering the weighted total group cost (\WTGC). Our result shows that the \ENDPOINT mechanism achieves an approximation ratio of $1+(n-2)\cdot \frac{w_{\max}}{w_{\min}}$.
\begin{theorem}
    The \ENDPOINT mechanism has an approximation ratio of $1 + (n-2)\frac{w_{\max}}{w_{\min}}$ for minimizing $\mge$ when $E_j = w_j \cdot \sum_{i \in G_j} c(f(\bftheta), x_i)$.
\end{theorem}
\begin{proof}
    Given any agent profile $\bftheta$, let $Y=(x_1,x_n)$ denote the outputs of the \ENDPOINT mechanism and $Y^*=(y_1^*,y_2^*)$ (w.l.o.g, $y^*_1\leq y^*_2$) denote the optimal facility locations which achieves optimal $\mge(\bftheta,Y^*)$. We first observe that $x_1\leq y^*_1 \leq y^*_2 \leq x_n$. 
    For any group $G_j$, suppose there are $k_1^j$ agents (excluding agent $1$) who are assigned to facility $y^*_1$ while $k_2^j$ agents (excluding agent $n$) assigned to facility $y^*_2$\footnote{Breaking ties by assigning to $y^*_1$}. Now we consider the follow movement, moving one facility from $y^*_1$ to $x_1$ and the other facility from $y^*_2$ to $x_n$. For group $G_j$, after the movement, the changes of the group effect is expressed as 
    \begin{align*}
        E_j(Y)\hspace{-.1em} -\hspace{-.1em} E_j(Y^*) &\leq w_j  \left(k_1^j(y^*_1-x_1)+k_2^j(x_n-y^*_2)\right) \\
        &\leq w_j(k_1^j+k_2^j)\hspace{-.1em}\cdot\hspace{-.1em} \max\{y^*_1-x_1,x_n-y^*_2\} \\
        & \leq w_j(n-2)\hspace{-.1em}\cdot\hspace{-.1em} \max\{y^*_1-x_1,x_n-y^*_2\}.
    \end{align*}
    Recall that $w_{\max} = \max_j w_j$ and $w_{\min} = \min_j w_j$. Now we consider the $\mge$ objective and have 
    \begin{align}
        \mge(\bftheta,Y)-\mge(\bftheta,Y^*)
        &\leq\max_{j\in [m]}(E_j(Y)-E_j(Y^*))\notag \\
       & \leq w_{\max}\cdot (n-2)\cdot \max\{y^*_1-x_1,x_n-y^*_2\}.\label{eq:2}
    \end{align}
    On the other hand, since there exists at least one agent who is assigned to each facility under the optimal solution $Y^*$, we have the lower bound that
    \begin{align}
        \mge(\bftheta,Y^*) &\geq \max\{w_{g_1}(y^*_1-x_1),w_{g_n}(x_n-y^*_2)\}\notag \\
        &\geq w_{\min}\cdot \max\{y^*_1-x_1, x_n-y^*_2\}.\label{eq:3}
    \end{align}
    From \Cref{eq:2} and \Cref{eq:3}, we derive the upper bound of the approximation ratio $\rho$
\begin{align*}
      \frac{\mge(\bftheta, Y)}{\mge(\bftheta, Y^*)}
      \le& 1+ \frac{(n-2)w_{\max}\hspace{-0.1em}\cdot\hspace{-0.1em} \max\left\{y_1^* - x_1, x_n - y_2^* \right\}}{\max\left\{w_{g_1}(y_1^* - x_1), w_{g_n}(x_n - y_2^*) \right\}}\\
       \le& 1+ \frac{(n-2)w_{\max}\hspace{-.1em}\cdot \hspace{-.1em}  \max\left\{y_1^* - x_1, x_n - y_2^* \right\}}{w_{\min} \max\left\{y_1^* - x_1, x_n - y_2^* \right\}}\\
       \le& 1+(n-2)\frac{w_{\max}}{w_{\min}}.
\end{align*}
To show the tightness, consider an instance with $n$ agents where $x_1=0,x_2=x_3=\dots=x_{n-1}=\frac{1}{2}$, and $x_n=1$, and $G_1=\{1\}$, $G_2=\{2,3,\dots,n\}$. For group weights, let $w_1=w_{\min}$ and $w_2=w_{\max}$. We first observe the optimal solution $Y^*=(y_1^*=\frac{w_{\max} (n-2)}{2 [w_{\min} + w_{\max} (n-2)]}, y_2^*=1)$, achieving $\mge(\bftheta,Y^*)=\frac{w_{\min} w_{\max} (n-2)}{2 [w_{\min} + w_{\max} (n-2)]}$. In contrast,  the \ENDPOINT mechanism has an $\mge$ of $\frac{w_{\max} (n-2)}{2}$. This gives us an approximation ratio of $1+(n-2)\cdot \frac{w_{\max}}{w_{\min}}$.
\end{proof}

We adapt the characterization of \citet{FotakisT14}, which identifies the \ENDPOINT mechanism as the unique deterministic, anonymous, and strategyproof mechanism with bounded approximation ratio for $k=2$, to establish the tightness.
\begin{proposition}\label{prop:lower_bound_wtgc}
    Any deterministic, strategyproof mechanism has an approximation ratio of at least $1+(n-2)\cdot \frac{w_{\max}}{w_{\min}}$ when $E_j = w_j \cdot \sum_{i \in G_j} c(f(\bftheta), x_i)$.
\end{proposition}

\subsection{Weighted Maximum Group Cost}
We now turn to the weighted maximum group cost ($\WMGC$) objective. Under this criterion, the \ENDPOINT mechanism attains an approximation ratio of ($1+\frac{w_{\max}}{w_{\min}}$).
\begin{theorem}
    \label{thm:endpoint_wmgc_upper_bound}
    The \ENDPOINT mechanism has an approximation ratio of $1 + \frac{w_{\max}}{w_{\min}}$ for minimizing $\mge$ when $E_j = w_j \cdot \max_{i \in G_j} c(f(\bftheta), x_i)$.
\end{theorem}

\begin{proof}
    Given any agent profile $\bftheta$, Denote by $Y=(x_1,x_n)$ the outputs of the \ENDPOINT mechanism and $Y^*=(y_1^*,y_2^*)$ ($y^*_1\leq y^*_2$) the optimal facility placement which achieves optimal $\mge(\bftheta,Y^*)$. We first observe that $x_1\leq y^*_1 \leq y^*_2 \leq x_n$. Without loss of generality, assume $\mge(\bftheta,Y)$ is achieved by agent $k$ and $x_k\leq \frac{x_1+x_n}{2}$, i.e., $k$ is assigned to facility located at $x_1$ under $Y$.

    \textbf{Case 1. When $x_1 \leq x_k \leq y_1^*$,} we have $\mge(\bftheta,Y)=w_{g_k}(x_k-x_1)$ and $\mge(\bftheta,Y^*)\geq w_{g_1}(y^*_1-x_1)$. The approximation ratio $\rho$ is upper-bounded by 
    \begin{align*}
    \frac{\mge(\bftheta, Y)}{\mge(\bftheta, Y^*)} 
    \hspace{-.1em}\le \hspace{-.1em}\frac{w_{g_k}\hspace{-.3em}\cdot\hspace{-.2em} (x_k-x_1)}{w_{g_1}\hspace{-.3em}\cdot\hspace{-.2em} (y^*_1-x_1)} \hspace{-.1em}\le \hspace{-.1em}\frac{w_{g_k}\hspace{-.3em}\cdot \hspace{-.2em}(x_k-x_1)}{w_{g_1}\hspace{-.3em}\cdot\hspace{-.2em} (x_k-x_1)} \hspace{-.1em}\le \hspace{-.1em}\frac{w_{\max}}{w_{\min}}.
    \end{align*}

    \textbf{Case 2. When $y^*_1<x_k<y^*_2$ and $k$ is assigned to $y_1^*$.} By \Cref{lem:two_agents}, we have $\mge(\bftheta,Y^*) \geq \frac{w_{g_1}w_{g_k}(x_k-x_1)}{w_{g_1}+w_{g_k}}$ and the approximation ratio $\rho$ is upper-bounded by 
    \begin{align*}
        \frac{\mge(\bftheta,Y)}{\mge(\bftheta,Y^*)}\leq \frac{w_{g_k}\cdot (x_k-x_1)}{\frac{w_{g_1}w_{g_k}(x_k-x_1)}{w_{g_1}+w_{g_k}}}\leq 1+\frac{w_{g_k}}{w_{g_1}}\leq 1+\frac{w_{\max}}{w_{\min}}.
    \end{align*}

    \textbf{Case 3. When $y^*_1<x_k<y^*_2$ and agent $k$ is assigned to $y_2^*$.} Similarly, we can lower-bound $\mge(\bftheta, Y^*) \geq \frac{w_{g_k}w_{g_n}(x_n-x_k)}{w_{g_k}+w_{g_n}}$. Recall that $x_k\leq \frac{x_1+x_n}{2}$. It implies that $x_k-x_1\leq x_n-x_k$. So we have the lower bound for the approximation ratio $\rho$ that
    \begin{align*}
    \frac{\mge(\bftheta, Y)}{\mge(\bftheta, Y^*)} \leq \frac{w_{g_k}\cdot (x_k-x_1)}{\frac{w_{g_k}w_{g_n}(x_n-x_k)}{w_{g_k}+w_{g_n}}}\leq 1 + \frac{w_{g_k}}{w_{g_n}}\leq 1 + \frac{w_{\max}}{w_{\min}}.
    \end{align*}

    \textbf{Case 4. When $y^*_2\leq x_k\leq x_n$.} In this case, agent $k$ is assigned to the facility located at $y^*_2$. Hence we derive that $\mge(\bftheta,Y^*) \geq w_n(x_n-y^*_2)$. Consequently, the approximation ratio $\rho$ satisfies 
    $\rho = \frac{\mge(\bftheta, Y)}{\mge(\bftheta, Y^*)} 
    \leq \frac{w_{g_k}\cdot (x_k-x_1)}{w_n \cdot (x_n-y_2^*)}$.
    Note that we also have $x_k-x_1 \leq x_n-x_k$ and $x_n-y^*_2 \geq x_n-x_k$. Therefore, we bound the approximation ratio $\rho$ by
    \begin{align*}
   \frac{\mge(\bftheta, Y)}{\mge(\bftheta, Y^*)} 
    \hspace{-.1em}\leq \hspace{-.1em}\frac{w_{g_k}\hspace{-.3em}\cdot \hspace{-.2em}(x_k-x_1)}{w_n \hspace{-.3em}\cdot \hspace{-.2em} (x_n-y_2^*)} \hspace{-.1em}\leq \hspace{-.1em} \frac{w_{g_k}\hspace{-.3em}\cdot \hspace{-.2em} (x_n-x_k)}{w_n \hspace{-.3em}\cdot \hspace{-.2em} (x_n-x_k)}\hspace{-.1em}\leq \hspace{-.1em}\frac{w_{\max}}{w_{\min}}. 
    \end{align*}
    Combining the analysis across all four cases, we conclude that the \ENDPOINT mechanism achieves an approximation ratio of $1+\frac{w_{\max}}{w_{\min}}$. To establish the tightness of this bound, consider the following instance. There are $n$
    agents where $x_1=x_2=\dots=x_{n-2}=0$, $x_{n-1}=\frac{1}{2}$, and $x_n=1$. The group structure is given by $G_1=\{1\}$, and $G_2=\{2,3,\dots,n\}$ with weights $w_1=w_{\min}$, and $w_2=w_{\max}$. We first identify that the optimal solution is $Y^*=(y^*_1=\frac{w_{\max}}{2(w_{\max}+w_{\min})},y^*_2=1)$, achieving an $\mge$ value of $\frac{w_{\min}\cdot w_{\max}}{2(w_{\min}+w_{\max})}$. In contrast, the \ENDPOINT mechanism attains an $\mge$ value of $\frac{w_{\max}}{2}$, implying that it has an approximation ratio of $1+\frac{w_{\max}}{w_{\min}}$.
\end{proof}

\begin{proposition}\label{prop:lower_bound_wmgc}
    Any deterministic, strategyproof mechanism has an approximation ratio of at least $1 + \frac{w_{\max}}{w_{\min}}$ when $E_j = w_j \cdot \max_{i \in G_j} c(f(\bftheta), x_i)$.
\end{proposition}

\section{Conclusion and Discussion}\label{sec:conclusion}
We study facility location games through the lens of fairness by introducing a unified framework based on the \emph{maximum group effect}, a general metric that encompasses a broad class of classical objectives. In the single-facility setting, we develop two strategyproof mechanisms, \BALANCED and \MajorPhantom, both of which achieve tight approximation guarantees for minimizing the maximum group effect. Our results recover the known tight bounds for the classical group-fair objectives as special cases and extend these guarantees to weighted and overlapping group structures. In the two-facility setting, we revisit the classical \ENDPOINT mechanism and establish tight approximation bounds. Looking forward, promising research directions include extending our framework to randomized mechanisms to circumvent the impossibility of achieving bounded approximations for $k\geq 3$, as well as adapting the $\mge$ objective to higher dimensional metric spaces.

\section*{Acknowledgements}
This work was supported by the NSF-CSIRO grant on “Fair Sequential Collective Decision-Making" (RG230833) and the ARC Laureate Project FL200100204 on “Trustworthy AI”. The authors would like to express their gratitude to the anonymous reviewers of AAAI 2026 for their insightful and constructive feedback, which greatly helped improve this paper.

\bibliographystyle{plainnat}
\bibliography{ref}

@inproceedings{ZLC22a,
  title={Strategyproof Mechanisms for Group-Fair Facility Location Problems},
  author={Zhou, Houyu and Li, Miniming and Chan, Hau},
  booktitle={Proceedings of the 31st International Joint Conference on Artificial Intelligence and the 25th European Conference on Artificial Intelligence (IJCAI)},
  pages={613--619},
  year={2022}
}

@inproceedings{ProcacciaT09,
  author       = {Ariel D. Procaccia and
                  Moshe Tennenholtz},
  title        = {Approximate mechanism design without money},
  booktitle    = {Proceedings of the 10th {ACM} Conference on Electronic Commerce (EC)},
  pages        = {177--186},
  year         = {2009}
}

@inproceedings{CFL+21a,
  author       = {Hau Chan and
                  Aris Filos{-}Ratsikas and
                  Bo Li and
                  Minming Li and
                  Chenhao Wang},
  title        = {Mechanism Design for Facility Location Problems: {A} Survey},
  booktitle    = {Proceedings of the 30th International Joint Conference on Artificial Intelligence (IJCAI)},
  pages        = {4356--4365},
  year         = {2021}
}

@article{MarshS94,
  title={Equity measurement in facility location analysis: A review and framework},
  author={Marsh, Michael T and Schilling, David A},
  journal={European journal of operational research},
  volume={74},
  number={1},
  pages={1--17},
  year={1994}
}

@inproceedings{Toby25,
  title     = {Equitable Mechanism Design for Facility Location},
  author    = {Walsh, Toby},
  booktitle = {Proceedings of the 34th International Joint Conference on
               Artificial Intelligence, (IJCAI)},
  pages     = {275--283},
  year      = {2025},
}

@article{Rawls1958,
  title={Justice as fairness},
  author={Rawls, John},
  journal={The philosophical review},
  volume={67},
  number={2},
  pages={164--194},
  year={1958}
}

@article{FotakisT14,
  author       = {Dimitris Fotakis and
                  Christos Tzamos},
  title        = {On the Power of Deterministic Mechanisms for Facility Location Games},
  journal      = {{ACM} Transactions on Economics and Computation},
  volume       = {2},
  number       = {4},
  pages        = {15:1--15:37},
  year         = {2014}
}

@inproceedings{cai2016envy,
author = {Cai, Qingpeng and Filos-Ratsikas, Aris and Tang, Pingzhong},
title = {Facility location with minimax envy},
year = {2016},
booktitle = {Proceedings of the 25th International Joint Conference on Artificial Intelligence (IJCAI)},
pages = {137–143}
}

@article{ALL+25a,
title = {Proportionality-based fairness and strategyproofness in the facility location problem},
journal = {Journal of Mathematical Economics},
volume = {119},
pages = {103129},
year = {2025},
author = {Haris Aziz and Alexander Lam and Barton E. Lee and Toby Walsh}
}

@article{McAl76a,
  title={Equity and efficiency in public facility location},
  author={McAllister, Donald M},
  journal={Geographical analysis},
  volume={8},
  number={1},
  pages={47--63},
  year={1976},
  publisher={Wiley Online Library}
}

@article{CFL+22a,
  title={Strategyproof mechanisms for 2-facility location games with minimax envy},
  author={Chen, Xin and Fang, Qizhi and Liu, Wenjing and Ding, Yuan and Nong, Qingqin},
  journal={Journal of Combinatorial Optimization},
  volume={43},
  number={5},
  pages={1628--1644},
  year={2022}
}

@inproceedings{LAL+24a,
  title={Proportional fairness in obnoxious facility location},
  author={Lam, Alexander and Aziz, Haris and Li, Bo and Ramezani, Fahimeh and Walsh, Toby},
  booktitle={Proceedings of the 23rd International Conference on Autonomous Agents and Multiagent Systems (AAMAS)},
  pages={1075--1083},
  year={2024}
}

@inproceedings{liu2020envyratio,
author = {Liu, Wenjing and Ding, Yuan and Chen, Xin and Fang, Qizhi and Nong, Qingqin},
title = {Multiple Facility Location Games with Envy Ratio},
year = {2020},
booktitle = {Proceedings of the 14th International Conference on Algorithmic Aspects in Information and Management (AAIM)},
pages = {248–259}
}

@article{DLC+20a,
title = {Facility location game with envy ratio},
journal = {Computers \& Industrial Engineering},
volume = {148},
pages = {106710},
year = {2020},
author = {Yuan Ding and Wenjing Liu and Xin Chen and Qizhi Fang and Qingqin Nong},
}

@article{Moulin80a,
  title={On strategy-proofness and single peakedness},
  author={Moulin, Herv{\'e}},
  journal={Public Choice},
  volume={35},
  number={4},
  pages={437--455},
  year={1980},
  publisher={Springer}
}

@inproceedings{LiLC24a,
  author       = {Jiaqian Li and
                  Minming Li and
                  Hau Chan},
  title        = {Strategyproof Mechanisms for Group-Fair Obnoxious Facility Location
                  Problems},
  booktitle    = {Proceedings of the 38th {AAAI} Conference on Artificial Intelligence, (AAAI)},
  pages        = {9832--9839},
  year         = {2024},
}

@inproceedings{AzizCSWX25,
  author       = {Haris Aziz and
                  Hau Chan and
                  Xingchen Sha and
                  Toby Walsh and
                  Lirong Xia},
  title        = {Group Fairness in Multi-period Mobile Facility Location Problems},
  booktitle    = {Proceedings of the 24th International Conference on Autonomous Agents and Multiagent Systems, (AAMAS)},
  pages        = {2414--2416},
  year         = {2025}
}

@inproceedings{ZhouCL24,
  author       = {Houyu Zhou and
                  Hau Chan and
                  Minming Li},
  title        = {Altruism in Facility Location Problems},
  booktitle    = {Thirty-Eighth {AAAI} Conference on Artificial Intelligence, (AAAI)},
  pages        = {9993--10001},
  year         = {2024}
}

@inproceedings{LiPWZ25,
  author       = {Minming Li and
                  Cheng Peng and
                  Ying Wang and
                  Houyu Zhou},
  title        = {Group-fair Facility Location Games with Externalities},
  booktitle    = {Proceedings of the 24th International Conference on Autonomous Agents and Multiagent Systems, (AAMAS), 2025},
  pages        = {2615--2617},
  year         = {2025}
}
\clearpage
\appendix
\section{Comparative Analysis}\label{sec:comparative}
In this section, we revisit three classical mechanisms introduced by \citet{ProcacciaT09} and \citet{ZLC22a}: the \MED mechanism, the \LEFTMOST mechanism, and the Majority Group Median (\MAJOR) Mechanism. We then analyze the approximation ratios of these mechanisms under the maximum group effect objective, showing that our proposed mechanisms, \BALANCED and \MajorPhantom, achieve superior performance in approximating optimum for both the $\WTGC$ and $\WMGC$ objectives. A summary of the comparative results is presented in \Cref{tab:compare}.
\begin{table}[htb]
\centering
\begin{NiceTabular}{@{}ccc@{}}
\toprule
\textbf{Objectives} & \textbf{Mechanisms} & \textbf{Bounds} \\
\midrule
\multirow{4}{*}{$\WTGC$} & \MED & $1+(m-1)\cdot \frac{w_{\max}}{w_{\min}}$ \\
 & \LEFTMOST & $1+(n-1)\cdot \frac{w_{\max}}{w_{\min}}$ \\
 & \MAJOR & $1+2\cdot \frac{w_{\max}}{w_{\min}}$\\
 & \cellcolor{gray!20}\BALANCED & \cellcolor{gray!20}$\mathbf{2}$\\ 
 \hline
\multirow{4}{*}{$\WMGC$} & \MED & \multirow{3}{*}{$1+\frac{w_{\max}}{w_{\min}}$} \\
 & \LEFTMOST & \\
 & \MAJOR & \\
 & \cellcolor{gray!20}\MajorPhantom & \cellcolor{gray!20}$\mathbf{2}$\\ 
\bottomrule
\end{NiceTabular}
\caption{Summary of the comparative analysis}
\label{tab:compare}
\end{table}

The \MED mechanism and \LEFTMOST mechanism are the fundamental mechanisms studied by \citet{ProcacciaT09}. Given any agent profile $\bftheta$, without loss of generality, we assume that $x_1 \leq x_2 \leq \dots \leq x_n$. The \MED mechanism places the facility at $f(\bftheta)=x_{\lceil \frac{n}{2} \rceil}$ and the \LEFTMOST mechanism places the facility at $f(\bftheta)=x_1$. 
Regarding the Majority Group Median (\MAJOR) Mechanism by \citet{ZLC22a}. It places the facility at the median point among the largest-size group agents. 

\begin{algorithm}[!htbp]
\caption{\MAJOR \textbf{Mechanism}}
\label{algo:major_group}
\begin{algorithmic}[1] 
\REQUIRE Agent profile $\bftheta$, group weights $\{w_j\}_{j\in [m]}$.
\STATE Let $G_{\ell}$ denote the largest size group and $\bfx_\ell= \{x^\ell_1, \dots, x^\ell_{|G_{\ell}|}\}$ be the location profile of agents, tie-breaking in favor of the smallest index. 
\STATE Locate the facility $f(\bftheta)$ at the median point of multiset $\bfx=\{x^\ell_1, \dots, x^\ell_{|G_{\ell}|}\}$, tie-breaking by select the leftmost.
\ENSURE Facility location $f(\bftheta)$.
\end{algorithmic}
\end{algorithm}

We next study the approximation performance for these three mechanisms in terms of both the weighted total group cost ($\WTGC$) objective and the weighted maximum group cost ($\WMGC$) objective when placing one single facility $(k=1)$.

\subsection{Weighted Total Group Cost}
We begin by considering the weighted group cost $(\WTGC)$ objective, defined as $ E_j = w_j \cdot \sum_{i \in G_j} c(f(\mathbf{\theta}), x_i)$. We slightly abuse notation by using the functions $L_j(y)$ and $R_j(y)$ to denote the number of agents located to the left of any location $y$ and to the right of $y$ for any group $G_j\in \mathcal{G}$, respectively. Throughout, let $w_{\max}$ and $w_{\min}$ denotes the maximum weight and the minimum group weights.

\begin{proposition}\label{prop:med_total}
    The \MED mechanism has an approximation ratio of $1+(m-1)\cdot \frac{w_{\max}}{w_{\min}}$ for minimizing $\mge$ when $E_j = w_j \cdot \sum_{i \in G_j} c(f(\bftheta), x_i)$.
\end{proposition}
\begin{proof}
Consider an arbitrary profile $\bftheta$. Let $f(\bftheta)$ be the output of the \MED mechanism and $y^*$ be the optimal facility location. Without loss of generality, assume that $f(\bftheta) < y^*$ and that $\mge(\bftheta, f(\bftheta))$ is attained by some group $G_g$. We first claim that, the value of $\mge(\bftheta, f(\bftheta))$ is non-decreasing, while $\mge(\bftheta, y^*)$ is non-increasing when all the agents located between $f(\bftheta)$ and $y^*$ are relocated to $y^*$. Consequently, the approximation ratio of $f(\bftheta)$ under $\bftheta$ is upper-bounded by that of the resulting instance after the relocation. Henceforth, we restrict attention to instances in which no agent is located strictly between $f(\bftheta)$ and $y^*$.

Note that $\mge(\bftheta, f(\bftheta))=w_{g}\sum_{i \in G_{g}}d(x_i,f(\bftheta))$ and by the definition $\mge(\bftheta, y^*)\geq w_{g}\sum_{i \in G_{g}} d(x_i,y^*)$. We derive the difference between the two terms by
\begin{align*}
\mge(\bftheta, f(\bftheta))-\mge(\bftheta, y^*)
 & \leq w_{g}\sum_{i \in G_{g}}d(x_i,f(\bftheta))-w_{g}\sum_{i \in G_{g}} d(x_i,y^*) \\
 & \le w_g(|R_{g}(f(\bftheta))|-|L_{g}(f(\bftheta))|)(y^*-f(\bftheta)).
\end{align*}
Observe that the $\mge$ objective is at least the sum of the costs of agents in the same group that are positioned to the left of $f(\bftheta)$. We have
\begin{align*}
 & \mge(\bftheta, y^*) \geq \max_{j\in [m]} w_j|L_j(f(\bftheta))| (y^*-f(\bftheta)).
\end{align*}

We next discuss by cases.
\begin{enumerate}
    \item \textbf{When} $L_{g}(f(\bftheta))=0$, there are at most $m-1$ groups on the left of $f(\bftheta)$. By the definition of the \MED mechanism, there are at least $\frac{n}2$ agents on the left of $f(\bftheta)$ and at most $\frac{n}{2}$ agents on the right of $f(\bftheta)$. Hence, we have $\max_{j\in [m]} |L_j(f(\bftheta))|\geq \frac{n}{2(m-1)}$ and $|R_g(f(\bftheta))| \le \frac{n}{2}$. Therefore, we can bound the approximation ratio by
    \begin{align*}
            \rho& =\frac{\mge(\bftheta, f(\bftheta))}{\mge(\bftheta, y^*)}\\
            & \leq \frac{\mge(\bftheta, y^*)+w_g\cdot (|R_{g}(f(\bftheta))|-|L_g(f(\bftheta)|)(y^*-f(\bftheta))}{\mge(\bftheta, y^*)}\\
            &\leq \frac{\mge(\bftheta, y^*)+w_g\cdot |R_{g}(f(\bftheta))|(y^*-f(\bftheta))}{\mge(\bftheta, y^*)}\\
            &\leq 1+\frac{w_g\cdot |R_{g}(f(\bftheta))|(y^*-f(\bftheta))}{\max_{j\in [m]} w_j \cdot |L_j(f(\bftheta))| (y^*-f(\bftheta))}\\
            & \leq 1+\frac{w_g\cdot \frac{n}{2}}{\max_{j\in [m]} w_j \cdot \frac{n}{2(m-1)}}\\
            &\leq 1 + \frac{w_g\cdot (m-1)}{\max_{j\in [m]}w_j}\\
            & \leq 1+(m-1)\cdot \frac{w_{\max}}{w_{\min}}.
    \end{align*}
    
   \item \textbf{When} $L_{g}(f(\bftheta))>0$, we first rewrite the approximation ratio by  
    \begin{align*}
    \rho = & \frac{\mge(\bftheta, f(\bftheta))}{\mge(\bftheta, y^*)} \\
    \leq & \frac{\mge(\bftheta, y^*)+w_g(|R_{g}(f(\bftheta))|-|L_{g}(f(\bftheta))|)(y^*-f(\bftheta))}{\mge(\bftheta, y^*)} \\ 
    \leq &1+\frac{w_g(|R_{g}(f(\bftheta))|-|L_{g}(f(\bftheta))|)(y^*-f(\bftheta))}{\max_{j\in [m]}w_j |L_j(f(\bftheta))| (y^*-f(\bftheta)).} 
    \end{align*}
    After that, we further consider two sub-cases by discussing whether the group $G_g$ achieves $\max_{j\in [m]}|L_j(f(\bftheta))|$. 
    \begin{enumerate}
        \item \textbf{When $\max_{j\in [m]} |L_j(f(\bftheta))|$ is achieved by group $g$.} By the definition of the \MED mechanism, there are at least $\frac{n}{2}$ agents on the left of $f(\bftheta)$ and at most $\frac{n}{2}$ agents on the right of $f(\bftheta)$. Since there are at most $m$ groups, by the pigeonhole principle,  we have $|L_{g}(f(\bftheta))| \geq \frac{\frac{n}{2}}{m}$. We have the approximation ratio
    \begin{align*}
    \rho & \leq 1+\frac{w_g(|R_{g}(f(\bftheta))|-|L_{g}(f(\bftheta))|)(y^*-f(\bftheta))}{w_g |L_j(f(\bftheta))| (y^*-f(\bftheta))}\\ 
    & =1+\frac{\left(|R_{g}(f(\bftheta))|-|L_{g}(f(\bftheta))|\right)}{|L_{g}(f(\bftheta))|}\\
    & =\frac{|R_{g}(f(\bftheta))|}{|L_{g}(f(\bftheta))|}\leq \frac{\frac{n}{2}}{\frac{\frac{n}{2}}{m}}=m.
     \end{align*}
    \item \textbf{When $\max_{j\in [m]} |L_j(f(\bftheta))|$ is not achieved by group $g$.} By the pigeonhole principle, we derive that
    \begin{align*}
        \max_{j\in [m]} |L_j(f(\bftheta))|\geq \frac{\frac{n}{2}-|L_{g}(f(\bftheta))|}{m-1}.
    \end{align*}
    With the inequality in hand, we upper-bound the approximation ratio by
    \begin{align*}
    \rho & \leq 1+\frac{w_g(|R_{g}(f(\bftheta))|-|L_{g}(f(\bftheta))|)(y^*-f(\bftheta))}{\max_{j\in [m]} w_j |L_j(f(\bftheta))| (y^*-f(\bftheta))}\\ 
    & \leq 1+\frac{w_g\left(|R_{g}(f(\bftheta))|-|L_{g}(f(\bftheta))|\right)}{\max_{j\in [m]} w_j\frac{\frac{n}{2}-|L_{g}(f(\bftheta))|}{m-1} }\\ 
    & \leq 1+\frac{w_{\max}}{w_{\min}}\cdot\frac{\frac{n}{2}-|L_{g}(f(\bftheta))|}{\frac{\frac{n}{2}-|L_{g}(f(\bftheta))|}{m-1} }\\
    & \leq 1+(m-1)\cdot \frac{w_{\max}}{w_{\min}}.
    \end{align*}
    \end{enumerate}
\end{enumerate}
Lastly, we provide an example to show that our analysis for the approximation ratio of \MED mechanism is tight. Consider an instance where $m-1$ agents from different groups $G_2,\dots, G_m$ with weight $w_{\min}$ are located at $0$ and the other $m-1$ agents in $G_1$ with weight $w_{\max}$ are located at $1$. The \MED mechanism locates the facility at $0$ achieving an $\mge$ of $(m-1)\cdot w_{\max}$, while locating the facility at $\frac{(m-1)\cdot w_{\max}}{(m-1)\cdot w_{\max}+w_{\min}}$ can achieve the optimal $\mge$ of $\frac{(m-1)\cdot w_{\max}\cdot w_{\min}}{(m-1)\cdot w_{\max}+w_{\min}}$. Hence the approximation ratio of the \MED mechanism is at least $1+(m-1)\cdot \frac{w_{\max}}{w_{\min}}$, which completes the proof.
\end{proof}

We next prove that the \LEFTMOST mechanism achieves an approximation ratio of $1+(n-1)\cdot \frac{w_{\max}}{w_{\min}}$ for minimizing the $\mge$ objective when $E_j = w_j \cdot \sum_{i \in G_j} c(f(\bftheta), x_i)$.
\begin{proposition}
    The \LEFTMOST mechanism has an approximation ratio of $1+(n-1)\cdot \frac{w_{\max}}{w_{\min}}$ for minimizing $\mge$ when $E_j = w_j \cdot \sum_{i \in G_j} c(f(\bftheta), x_i)$.
\end{proposition}
\begin{proof}
Consider any profile $\bftheta$. Let $f(\bftheta)$ be the output of the \LEFTMOST mechanism and $y^*$ be the optimal facility location that minimizes $\mge$ objective. Without loss of generality, assume that $\mge(\bftheta, f(\bftheta))$ is achieved by group $G_g$. By definition of the \LEFTMOST mechanism, we first have $f(\bftheta) \leq y^*$. We use the analogous relocation approach as that in the proof of \Cref{prop:med_total} and focus on the instances after relocation. Note that $\mge(\bftheta, f(\bftheta))=w_g\sum_{i \in G_{g}}d(x_i,f(\bftheta))$ and $ \mge(\bftheta, y^*) \ge w_g \sum_{i \in G_{g}} d(x_i,y^*).$ Similar to the proof of \Cref{prop:med_total}, we have
\begin{align*}
 \mge(\bftheta, f(\bftheta))-\mge(\bftheta, y^*)
 \le &~ w_{g}(|R_{g}(f(\bftheta))|-|L_{g}(f(\bftheta))|)(y^*-f(\bftheta)),
\end{align*}
and since there are $|L_j(f(\bftheta))|$ agents in $G_j$ at or on the left of $f(\bftheta)$, we have
\begin{align*}
     \mge(\bftheta, y^*) \geq \max_{j\in [m]} w_j |L_j(f(\bftheta))| (y^*-f(\bftheta)).
\end{align*}
So the approximation ratio is bounded by
\begin{align*}
    \rho=&\frac{\mge(\bftheta, f(\bftheta))}{\mge(\bftheta, y^*)}\\
    \le & 1+\frac{w_g(|R_{g}(f(\bftheta))|-|L_{g}(f(\bftheta))|)(y^*-f(\bftheta))}{\max_{j} w_j |L_j(f(\bftheta))| (y^*-f(\bftheta))}\\
     =& 1+\frac{w_{\max}}{w_{\min}}\cdot \frac{|R_{g}(f(\bftheta))|-|L_{g}(f(\bftheta))|}{\max_{j} |L_j(f(\bftheta))| }.
\end{align*}
Notice that $\max_{j\in [m]} |L_j(y)|\geq 1$ and $|R_{g}(y)|-|L_{g}(f(\bftheta))|\leq n-1$. The approximation ratio is further bounded by 
\begin{align*}
    \rho & = 1+\frac{w_{\max}}{w_{\min}}\cdot \frac{|R_{g}(f(\bftheta))|-|L_{g}(f(\bftheta))|}{\max_{j} |L_j(f(\bftheta))| }\\
    &\leq 1+(n-1)\cdot \frac{w_{\max}}{w_{\min}}.
\end{align*}

The following example shows that our analysis for the approximation ratio of the \LEFTMOST mechanism is tight. Consider an instance with two groups where one agent in $G_1$ with $w_1=w_{\min}$ at $0$, and $n-1$ agents in $G_2$ with $w_2=w_{\max}$ at $1$. The optimal $\mge$ solution locates the facility at $\frac{(n-1)\cdot w_{\max}}{(n-1)\cdot w_{\max}+w_{\min}}$, achieving an $\mge$ of $\frac{(n-1)w_{\max}\cdot w_{\min}}{(n-1)w_{\max}+w_{\min}}$, and the \LEFTMOST mechanism locates the facility at $0$, achieving an $\mge$ of $(n-1)\cdot w_{\max}$. Hence, the approximate ratio of the \LEFTMOST mechanism is at least $1+(n-1)\cdot \frac{w_{\max}}{w_{\min}}$.
\end{proof}

We next turn our attention to the \MAJOR mechanism which was proposed by \citet{ZLC22a}. Under the $\mge$ objective, they showed that when $w_j = 1$ (i.e., maximum total group cost, $\mtgc$) or $w_j = \frac{1}{|G_j|}$ (i.e., maximum average group cost, $\magc$), the \MAJOR mechanism achieves an approximation ratio of $3$ in both cases. However, we show that the \MAJOR mechanism is not robust: when arbitrary group weights $w_j$ are allowed under the $\mge$ objective, it no longer guarantees a $3$-approximation.
\begin{proposition}
    The \MAJOR mechanism has an approximation ratio of $1 + 2\cdot \frac{w_{\max}}{w_{\min}}$ for minimizing $\mge$ objective when $E_j = w_j \cdot \sum_{i \in G_j} c(f(\bftheta), x_i)$.
\end{proposition}
\begin{proof}
Consider any arbitrary profile $\bftheta$. Let $f(\bftheta)$ be the output of the \MAJOR mechanism and $y^*$ be the optimal location. Without loss of generality, assume that $f(\bftheta) < y^*$ and $\mge(\bftheta, f(\bftheta))$ is achieved by group $G_g$. Recall \Cref{algo:major_group} and the group $G_\ell$ with largest size of agents. We first observe that there are at least $\frac{|G_\ell|}{2}$ agents in group $G_\ell$ on the left side of $f(\bftheta)$, which gives us  
\begin{align*}
    \mge(\bftheta, y^*) \ge w_\ell\cdot \frac{|G_\ell|}{2}\cdot (y^*-f(\bftheta)).
\end{align*}
Similar to the proof of \Cref{prop:med_total}, we have
\begin{align*}
 \mge(\bftheta, f(\bftheta))-\mge(\bftheta, y^*)
 &\le w_{g}(|R_{g}(f(\bftheta))|-|L_{g}(\bftheta)|)(y^*-f(\bftheta)).
\end{align*}
Finally, the approximation ratio is bounded by 
\begin{align*}
\rho & = \frac{\mge(\bftheta, f(\bftheta))}{\mge(\bftheta, y^*)} \\
& \le 1+\frac{w_{g}(|R_{g}(f(\bftheta))|-|L_{g}(\bftheta)|)(y^*-f(\bftheta))}{\mge(\bftheta, y^*)}\\
& \leq 1+\frac{w_g|G_g|(y^* - f(\bftheta))}{\mge(\bftheta, y^*)}\\
& \leq 1 + \frac{w_g \cdot |G_g|\cdot (y^* - f(\bftheta))}{w_\ell \cdot \frac{|G_\ell|}{2}\cdot (y^* - f(\bftheta))}\\
& \leq 1+2\cdot \frac{w_g}{w_\ell}\cdot \frac{|G_g|}{|G_\ell|} \\
& \leq 1+2\cdot \frac{w_{\max}}{w_{min}}.
\end{align*}
When considering $w_j=1$ or $w_j=\frac{1}{|G_j|}$, we have $w_{\max}=w_{\min}$ which implies $1+2\cdot \frac{w_{\max}}{w_{\min}}$ degenerates to an approximation ratio of $3$. We provide the following example that shows our analysis for the approximation ratio of the \MAJOR mechanism is tight. Consider an instance with one agent in group $G_1$ with $w_{\min}$ at $0$, one agent in group $G_1$ at $\frac{2\cdot w_{\max}}{2\cdot w_{\max}+w_{\min}}$, and two agents in group $G_2$ with $w_{\max}$ at $1$. The \MAJOR mechanism locates the facility at $0$, achieving the $\mge$ of $2\cdots w_{\max}$, however, locating the facility at $\frac{2w_{\max}}{2w_{\max}+w_{\min}}$ achieves the optimal $\mge$ value at $\frac{2w_{\max}\cdot w_{\min}}{2w_{\max}+w_{\min}}$. Hence, the approximation ratio of the \MAJOR mechanism is at least $1+2\frac{w_{\max}}{w_{min}}$.
\end{proof}

\subsection{Weighted Maximum Group Cost}
In this subsection, we consider maximum group effect defined by $E_j = w_j \cdot \max_{i \in G_j} c(f(\mathbf{\theta}), x_i)$. We first show that locating the facility at any agent location can achieve an approximation ratio of $1+\frac{w_{\max}}{w_{\min}}$.

\begin{lemma}\label{lem:compare_max}
    Given any agent profile $\bftheta$,
    Locating the facility at any agent $i$'s location $x_i$ achieves an approximation ratio of $1+\frac{w_{\max}}{w_{\min}}$.
\end{lemma}
\begin{proof}
    Given any profile $\bftheta$, without loss of generality, assume that the facility is located at $x_k$, and $\mge(\bftheta, x_k)$ is incurred by agent $\ell$, i.e., $\mge(\bftheta, x_k) = w_{g_\ell} \cdot d(x_\ell, x_k)$\footnote{Recall that $w_{g_i}$ denotes the maximum weight among all he groups to which agent $i$ belongs, i.e., $w_{g_i}=\max_{g\in g_i}w_g$.}. Let $y^*$ be the optimal facility location. Without loss of generality, we assume that $x_k \leq y^*$. 
    
    \textbf{Case (1). If $x_\ell \le  x_k \leq y^*$},
     we have
     \begin{align*}
        \mge(\bftheta, y^*) \ge w_{g_k} \cdot d(x_k, y^*) \ge w_{g_k} \cdot d(x_\ell, x_k).
     \end{align*}
    So the approximation ratio can be upper-bounded by
    \begin{align*}
        \rho = \frac{\mge(\bftheta, x_k)}{\mge(\bftheta, y^*)} \le \frac{w_{g_\ell}\cdot d(x_\ell, x_k)}{w_{g_k} \cdot d(x_\ell, x_k)} \le \frac{w_{\max}}{w_{\min}}.
    \end{align*}

    \textbf{Case (2). If $x_k < x_\ell \le y^*$}, we have
    \begin{align*}
        \mge(\bftheta, y^*) \ge w_{g_\ell} \cdot d(x_\ell, y^*) \ge w_{g_\ell} \cdot d(x_\ell, x_k).
    \end{align*}
    which implies an optimum.

    \textbf{Case (3). If $x_k < y^* < x_\ell$}, from \Cref{lem:two_agents}, we derive that
    \begin{align*}
        \mge(\bftheta, y^*) & \ge \max_{k,\ell}\left\{ w_{g_k}\cdot d(x_k, y^*), w_{g_\ell}\cdot d(x_\ell, y^*)\right\} \\
        & \geq \frac{w_{g_k} \cdot w_{g_\ell} \cdot d(x_\ell,x_k)}{w_{g_k} + w_{g_\ell}}.
    \end{align*}
    Consequently, we further bound the approximation ratio by 
    \begin{align*}
        \rho = \frac{\mge(\bftheta, x_k)}{\mge(\bftheta, y^*)} \leq \frac{w_{g_\ell} \cdot d(x_\ell, x_k)}{\frac{w_{g_k} \cdot w_{g_\ell} \cdot d(x_\ell,x_k)}{w_{g_k} + w_{g_\ell}}}  \leq \frac{w_{g_k}+w_{g_\ell}}{w_{g_k}} \leq 1+\frac{w_{\max}}{w_{\min}}.
    \end{align*}
\end{proof}
\begin{proposition}\label{prop:med_max}
    The \MED, \LEFTMOST, and \MAJOR mechanisms have an approximation ratio of $1+\frac{w_{\max}}{w_{\min}}$ for minimizing $\mge$ when $E_j = w_j \cdot \max_{i \in G_j} c(f(\bftheta), x_i)$.
\end{proposition}
\begin{proof}
    By \Cref{lem:compare_max}, we immediately have that the \MED, \LEFTMOST, and \MAJOR mechanisms have an approximation ratio of at most $1+\frac{w_{\max}}{w_{\min}}$. To show the tightness, we then show that the approximation ratio is at least $1+\frac{w_{\max}}{w_{\min}}$. Consider an example with one agent in group $G_1$ with $w_1=w_{\min}$ at $0$ and one agent in group $G_2$ with $w_2=w_{\max}$ at $1$. For any of these three mechanisms, locating the facility at $0$ achieves the $\mge$ objective of $w_{\max}$. For the optimal location $\frac{w_{\max}}{w_{\max} + w_{\min}}$, it achieves the $\mge$ objective of $\frac{w_{\max}\cdot w_{\min}}{w_{\max} + w_{\min}}$. 
\end{proof}

\section{Omitted Proofs from Section 3}
\subsection{Omitted Proof of \Cref{prop:balance_mechanism_mid_point_and_left_most}}
\begin{proof}
    For $m=1$ and equal group weight, $\mge$ degenerates to the classical social cost objective. In this case, \BALANCED mechanism places the facility at $\min \{y : L_1(y)\geq R_1(y)\}$, which is the leftmost location $y$ such that the number of agents to the left of $y$ is at least as large as the number of agents to the right. This is precisely the definition of median point among the $n$ agents\footnote{Breaking ties by selecting the left median when $n$ is even.}. For the maximum cost objective (i.e., $m=n$ and equal weights), the leftmost agent location will be selected by \BALANCED mechanism since we have $L_1(x_1)=1$ and $\max_{j\in [m]}R_j(x_1)=1$. 
\end{proof}

\subsection{Omitted Proof of \Cref{lem:two_agents}}
\begin{proof}
Given any two-agent instance profile $\bftheta$, we first observe that $y^*\in [x_1, x_2]$.

\textbf{Case 1: $g_1=g_2$, i.e., two agents with the same group membership}. Given any facility location $y$, \(\mge(\bftheta, y) = w_1 \cdot \max \{ |y - x_1|, |y - x_2| \}\). For \(x_1 < x_2\), the function \(\max \{ |y - x_1|, |y - x_2| \}\) is piecewise linear and convex, minimized where \(|y - x_1| = |y - x_2|\). By solving $y - x_1 = x_2 - y$, we have $y^* = \frac{x_1 + x_2}{2}$. When $y^* = \frac{x_1 + x_2}{2}$, the maximum group effect under $y^\ast$ is 
\begin{align*}
    \mge(\bftheta, y^*) = w_1 \cdot \frac{x_2 - x_1}{2}.
\end{align*}

\textbf{Case 2: $g_1\neq g_2$, i.e., two agents with different group membership}. Given any facility location $y$, \(\mge(\bftheta, y) = \max \{ w_1 |y - x_1|, w_2 |y - x_2| \}\). For \(x_1 < x_2\), minimize the convex function \(\max \{ w_1 |y - x_1|, w_2 |y - x_2| \}\). The minimum occurs where $w_1 |y - x_1| = w_2 |y - x_2|$. That is 
\begin{align*}
    w_1 (y - x_1) = w_2 (x_2 - y) \implies  y^* = \frac{w_2 x_2 + w_1 x_1}{w_1 + w_2}.
\end{align*}
At $y^*$:
\begin{align*}
    w_1 |y^* - x_1| = w_1 \left( \frac{w_2 x_2 + w_1 x_1}{w_1 + w_2} - x_1 \right) = \frac{w_1 w_2 (x_2 - x_1)}{w_1 + w_2},
\end{align*}
and 
\begin{align*}
    w_2 |y^* - x_2| = w_2 \left( x_2 - \frac{w_2 x_2 + w_1 x_1}{w_1 + w_2} \right) = \frac{w_1 w_2 (x_2 - x_1)}{w_1 + w_2}.
\end{align*}
Therefore, we have $\mge(\bftheta, y^*) = \frac{w_1 w_2 (x_2 - x_1)}{w_1 + w_2}$. This completes the proof. 
\end{proof} 

\section{Omitted Proofs from Section 4}
\subsection{Omitted Proof of \Cref{prop:lower_bound_wtgc}}
Observe that the objective $\mge$, where $E_j=w_j\cdot \sum_{i\in G_j} c(f(\bftheta,x_i))$, strictly generalizes the classical social cost objective as a special case. We recall the characterization of deterministic strategyproof mechanisms for the case $k=2$ due to \citet{FotakisT14}. 
Specifically, when the number of the agents $n\geq 5$, a deterministic strategyproof mechanism $f$ admits a bounded approximation ratio with respect to the social cost objective if and only if $f$ is either the \ENDPOINT mechanism or a \DICTATORIAL mechanism. For the \ENDPOINT mechanism, the approximation ratio is known to be $1+(n-2)\cdot \frac{w_{\max}}{w_{\min}}$. We next analyze the \DICTATORIAL mechanism. Let $j$ be the dictator agent, so that one facility is placed at the dictator's reported location $x_j$. For the location of the other facility, the \DICTATORIAL mechanism first considers the distance of $x_j$ to the leftmost and to the rightmost location, $d_l=|\min \bfx -x_j|$ and $d_r=|\max \bfx - x_j|$, respectively. The second facility is then placed at $x_j - \max\{d_l, 2d_r\}$ if $d_l > d_r$, and $x_j + \max\{d_r,2d_l\}$, otherwise. Consider the following instance with $n$ agents. The dictator agent $j$ is located at position $0$. There are $n-2$ agents located at position $\frac{1}{2}$, and one additional agent located at position $1$. The agents are partitioned into two groups as follows: $G_1=\{j\}$ with weight $w_1=w_{\min}$, and $G_2=\{1,2,\cdots, j-1, j+1,\cdots, n\}$ with weight $w_2=w_{\max}$. Under this instance, the \DICTATORIAL mechanism outputs the facility locations $f(\bftheta)=\{x_j=0,1\}$. This outcome coincides with that of the \ENDPOINT mechanism on the same instance. Consequently, the approximation ratio attained by the \DICTATORIAL mechanism in this case is also $1+(n-2)\cdot \frac{w_{\max}}{w_{\min}}$.

\subsection{Omitted Proof of \Cref{prop:lower_bound_wmgc}}
Observe that the objective $\mge$, where $E_j=w_j\cdot \max_{i\in G_j} c(f(\bftheta,x_i))$, strictly generalizes the classical maximum cost objective as a special case. We show that for $k \ge 2$ cases, every deterministic strategyproof mechanism with a bounded approximation ratio for minimizing the maximum cost have a bounded approximation ratio for minimizing the social cost. To see this, consider a deterministic strategyproof mechanism $f$ with an unbounded approximation ratio for minimizing the social cost. It implies that there exists an instance that the optimal social cost is $0$ and the social cost achieved by $f$ is larger than $0$. It further implies that the optimal maximum cost is $0$ and the maximum cost achieved by $f$ is larger than $0$, under the same instance. Hence, it is sufficient to consider all deterministic strategyproof mechanisms with bounded approximation ratios for minimizing the social cost, which are characterized by \citet{FotakisT14}. 
Specifically, when the number of the agents $n\geq 5$, a deterministic strategyproof mechanism $f$ admits a bounded approximation ratio with respect to the social cost objective if and only if $f$ is either the \ENDPOINT mechanism or a \DICTATORIAL mechanism. Therefore, for the maximum cost objective, we only focus on the \ENDPOINT mechanism and the \DICTATORIAL mechanism. As we have shown in \Cref{thm:endpoint_wmgc_upper_bound}, the \ENDPOINT mechanism achieves a tight approximation ratio of $1+\frac{w_{\max}}{w_{\min}}$ for the $\WMGC$ objective. To complete the proof of the existence lower bound, it suffices to show that the approximation ratio of the \DICTATORIAL mechanism is at least $1+\frac{w_{\max}}{w_{\min}}$. To see this, consider $j$ is the dictator, and the instance is as follows. $x_1=x_2=\cdots =x_{n-2}=0, x_{n-1}=\frac{1}{2}$, and $x_j=1$. The group structure is given by $G_1=\{1\}$, and $G_2=\{2,3,\cdots, n-1, j\}$, with weights $w_1=w_{\min}$ and $w_2=w_{\max}$. By the definition of the \DICTATORIAL mechanism, one facility is placed at $1$ and the other facility is placed at $0$. The outcome aligns with the \ENDPOINT mechanism and the instance is, in fact, the exact instance which shows the tight bound of $1+\frac{w_{\max}}{w_{\min}}$ for the mechanism, hence, it implies that the \DICTATORIAL mechanism has an approximation ratio of at least $1+\frac{w_{\max}}{w_{\min}}$. This completes the proof.
\end{document}